\documentclass[letterpaper,11pt]{article}

\usepackage{amsmath, amsthm, amssymb, fullpage}

\usepackage[linesnumbered,ruled,vlined,boxed]{algorithm2e}

\usepackage{algorithmic}
\usepackage{epsfig}
\usepackage{graphicx}
\usepackage{graphics}
\usepackage{url}
\usepackage{verbatim}

\newtheorem{theorem}{Theorem}[section]
\newtheorem{definition}{Definition}

\newtheorem{corollary}{Corollary}

\newtheorem{lemma}{Lemma}

\newtheorem{example}{Example}

\title{Unit-sphere games}

\author{Pingzhong Tang\\
IIIS, Tsinghua University\\
Beijing, China\\
\url{kenshinping@gmail.com}\\
\and Hanrui Zhang\\ 
IIIS, Tsinghua University\\
Beijing, China\\
\url{segtree@gmail.com}}

\begin{document}
\maketitle

\begin{abstract}

This paper introduces a class of games, called unit-sphere games, where strategies are real vectors with unit 2-norms (or, on a unit-sphere). As a result, they can no longer be interpreted as probability distributions over actions, but rather be thought of as allocations of one unit of resource to actions and the multiplicative payoff effect on each action is proportional to square-root of the amount of resource allocated to that action. The new definition generates a number of interesting consequences. We first characterize sufficient and necessary conditions under which a two-player unit-sphere game has a Nash equilibrium. The characterization effectively reduces solving a unit-sphere game to finding all eigenvalues and eigenvectors of the product of individual payoff matrices. For any unit-sphere game with non-negative payoff matrices, there always exists a unique Nash equilibrium; furthermore, the unique equilibrium is efficiently reachable via Cournot adjustment. In addition, we show that any equilibrium in positive unit-sphere games corresponds to approximate equilibria in the corresponding normal-form games. Analogous but weaker results are extended to positive $n$-player unit-sphere games. 

\end{abstract}

\section{Introduction}

Consider the following two games.
\begin{example}
{\em Protecting Manhattan.} Two police stations try to protect Manhattan (could be visualized as a rectangle) from a two-dimensional terrorist attack. Station A is responsible for protecting all the {\em streets}, i.e., the horizontal paths across the rectangle; while station B is responsible for protecting all the {\em avenues}, i.e., the vertical paths. Each police station has one unit of police force, distributes optimally its force among its paths, and derives a positive utility $u_{ij}^A$ (resp. $u_{ij}^B$) from successfully protecting each {\em subway station $S_{ij}$}, namely, the intersections of street $i$ and avenue $j$. The probability of successfully protecting a subway station is $\sqrt{a_ib_j}$, where $a_i$ and $b_j$ are the amount of police force station A and $B$ allocates to street $i$ and avenue $j$ respectively.
\end{example}

\begin{example}
{\em Graphical ads for combinatorial queries.}
Consider a situation where a user submits a query (say, ``Yellow Stone national park'') to a travel website. The query triggers interests from two complementary advertisement agencies: one has a collection of hotel ads and the other airline ads.  The website allocates two regions on its homepage, each with one unit of area, to the two groups of ads respectively. Agency A, with a set of hotel ads, tries to fill in the first region with graphical ads of hotels, similar for agency B and airline ads for the second region. Each agency derives positive utility $u_{ij}^A$ (resp. $u_{ij}^B$) if the user successfully purchases a combination of (hotel $i$, airline $j$). Note that here agency A's utility may depend on $j$ since payment rule may involve both $i$ and $j$. The probability that the user purchases the combination is $\sqrt{a_ib_j}$, where $a_i$ and $b_j$ are the areas $A$ and $B$ allocate to hotel ad $i$ and airline ad $j$ respectively (so $\sqrt{a_i}$ is roughly the height or width of ad $i$).
\end{example}

At first glance, both games resemble a general version of Blotto game~\cite{Blotto2006}, thus nontrivial to solve. However, a close scrutiny reveals an interesting pattern: both games be modeled as normal-form games where players need to allocate one unit of resources to actions and the payoff effect on each action is proportional to the square-root of amount of resource allocated to that action. As we shall see, both examples are instances of {\em positive unit-sphere games}, which possess unique, learnable pure Nash equilibria.

\section{Unit-sphere games}
    
Most of the paper deals with {\em 2-player unit-sphere games}. In Section~\ref{sec:n-player}, this definition is extended to accommodate any number of players.
        \begin{definition}
            A two-player {\em unit-sphere game (USG)} is defined by two matrices $A\times B$, where
            \begin{itemize}
                            \item $A$ is an $m\times n$ payoff matrix for player 1,
                            \item $B$ is an $n\times m$ payoff matrix for player 2.
            \end{itemize}          
        \end{definition}
   
A unit-sphere strategy $x$ for player 1 is a column vector of real numbers such that $ x \in \mathbb{R}^m,\ \|x\|_2 = 1$, while a strategy $y$ for player 2 is a column vector of real numbers such that $ x \in \mathbb{R}^n,\ \|y\|_2 = 1$. Given a {\em strategy profile} $(x,y)$, the utility obtained by player 1 is $x^TAy$ while the utility of player 2 is $y^TBx$. Other game-theoretical notions, such as best response and Nash equilibrium, follow standard definitions.
 
Mathematically, the above definition is a 2-player normal-form game except for the definition of strategy, where the restriction of unit $L_1$-norm is now replaced by unit $L_2$-norm. In other words, each unit-sphere strategy is a point on a unit sphere, rather than a probability distribution. This also implies that both $x$ and $y$ can be negative on some coordinates, as long as they are on a unit sphere.

It is important to note that a USG can just be thought of as a standard normal-form game where each pure strategy corresponds to a unit-sphere strategy and there are infinite many such strategies. From this perspective, the characterization theorems (Theorems~\ref{thm:1},~\ref{thm:3.3},~\ref{thm:3.4}) are sufficient and necessary conditions for a large class of games to have (unique) pure Nash equilibria. 

In this paper, we do not consider randomized unit-sphere strategies, for the following reasons. First of all, a randomization over unit-sphere strategies is no longer a unit-sphere strategy, thus not well-defined under our new definition. Secondly, it is not hard to see that such a randomized strategy has a $L_2$-norm less than 1 and is always utility-dominated by some unit-sphere strategy. Last but not least, we are interested in comparing unit-sphere strategy (which is somewhat mixed) to standard mixed strategy, in terms of existence and computation efficiency of Nash equilibrium. Adding another level of mixture makes the comparison less interesting.

One can also view players in a USG as {\em risk averse agents} whose payoffs, when facing a lottery outcome, are not linear expectations of their utilities on deterministic outcomes in the lottery, but concave expectations (in our case, a square-root function). In general, games with concave utility agents possess a mixed Nash equilibrium and it is in general {\sc PPAD-hard} to compute such an equilibrium~\cite[Theorem 1]{Fiat10}. Our model and results does not follow from Fiat and Papadimitriou in that, firstly, we allow for negative strategies, i.e., $x$ and $y$ can have negative entries, thus the whole strategy set is not necessarily convex, precluding a Nash style proof; secondly, when restricting to non-negative strategies, under the additional assumption of positive payoff matrices, we are able to show that a {\em unique} Nash equilibrium exists and easy to compute. Readers are referred to~\cite{Fiat10} and the references therein for an introduction on non-linear expectations.

Finally, in our definition, adding a positive constant to each payoff matrix cell no longer yields an equivalent USG. Intuitively, when adding a large constant to a player's payoff function, the player has more incentive to distribute her resource evenly among pure strategies. So, it loses generality to restrictions on positive payoff matrices. On the other hand, USGs are scale-invariant in the sense that multiplying a constant to a player's payoff function yields an equivalent USG.

\section{Nash Equilibria in USGs}
\label{sec:NE}

In this section, we characterize sufficient and necessary conditions for Nash equilibrium (NE) to exist in USGs. In particular, equilibrium exists in all the USGs with positive payoff matrices. It is unique and efficiently computable, via a well-known learning process known as {\em Cournot adjustment}. 
    

\subsection{Structure of NE in USGs}
        
            Let us now consider NE in a USG $A\times B$. It is easy to see that the utilities of the two players are
            $$
                u_1 = x^T A y = \|Ay\|_2 \cos \alpha,
            $$
            $$
                u_2 = y^T B x = \|Bx\|_2 \cos \beta,
            $$
 respectively, where $\alpha$ denotes the angle between $x$ and $Ay$ and $\beta$ denotes the angle between $y$ and $Bx$.\footnote{When $Ay = 0$ (resp. $Bx = 0$), one may set $\alpha$ (resp. $\beta$) arbitrarily.} Since both $x$ and $y$ are on the unit-sphere, a strategy profile $(x, y)$ forms an NE if and only if
            $$
                x = \mathrm{arg} \max_{x'} x'^T A y \iff \alpha = 0 \iff \lambda x = A y,
            $$
            and
            $$
                y = \mathrm{arg} \max_{y'} y'^T B x \iff \beta = 0 \iff \mu y = B x,
            $$
            where $\lambda = \|A y\|_2$, $\mu = \|B x\|_2$.

            By this observation, we have a necessary condition of existence of NE for two-player USGs.
            \begin{lemma}
            \label{lemma:1}
                Let $A$ and $B$ be the matrices of a USG. If $AB$ and $BA$ do not share a nonnegative eigenvalue, the USG does not have an NE.
            \end{lemma}
            \begin{proof}
                We show that an NE exists only if $AB$ and $BA$ share a nonnegative eigenvalue. Consider payoff matrices $A$ and $B$. For a NE profile $(x, y)$,
                $$
                    B \lambda x = B A y \Rightarrow \lambda \mu y = B A y,
                $$
                $$
                    A \lambda y = A B x \Rightarrow \lambda \mu x = A B x.
                $$
                In other words, $x$ is an eigenvector of $AB$ with eigenvalue $\lambda \mu$, and $y$ is an eigenvector of $BA$ with eigenvalue $\lambda \mu$.
            \end{proof}
            
            Since $AB$ and $BA$ have the same set of eigenvalues, the following theorem characterizes the sufficient and necessary condition for an NE to exist in any two-player USG.

            \begin{theorem}
            \label{thm:1}
                Let $A$ and $B$ be the matrices of an USG. There exists an NE for the USG if and only if $AB$ (or $BA$) has a nonnegative eigenvalue $\lambda \ge 0$.
            \end{theorem}
            \begin{proof}
                The only-if direction follows from Lemma~\ref{lemma:1}. We now prove the if direction. Assume $AB$ has a nonnegative eigenvalue $\lambda$ with eigenvector $x$ such that $\|x\|_2 = 1$.
                \begin{itemize}
                \item If $Bx \ne 0$, let $y = \frac{Bx}{\|Bx\|_2}$. $(x, y)$ is an NE for the game, because $\frac{\lambda}{\|Bx\|_2} x = A y$, and $\|Bx\|_2 y = Bx$.
                \item If $Bx = 0$, $y \ne 0$ can be chosen such that either $Ay = k x$ for some $k > 0$, when $\det A \ne 0$, or $Ay = 0$, when $\det A = 0$. Also we assume $\|y\|_2 = 1$. Again $(x, y)$ is an NE for the game, because $k x = Ay$ for some $k \ge 0$, and the utility of player 2, $y^T B x$, is always $0$.
                \end{itemize}
            \end{proof}
    
As stated in Theorem~\ref{thm:1}, to solve an USG $A\times B$, i.e., to find all NEs or to ensure that no NE exists, it is equivalent to calculate all eigenvalues of $AB$ and the corresponding eigenvectors. Solving USGs is therefore effectively reduced to the {\em eigenvalue problem}, for which one may refer to the standard {\em Singular value decomposition}. We refer readers to \cite{S02} for more efficient algorithms.
            
    \subsection{Positive USGs}
    
    We now focus on a general class of USGs where there always exists a unique NE.
        
        \begin{definition}
            A USG $A\times B$ is {\em positive} if $A , B > 0$, and any strategy satisfies $x, y \ge 0$\footnote{We say a matrix $A > 0$ if $A_{ij} > 0$ for all $(i,j)$, and a vector $x \ge 0$ if $x_i \ge 0$ for all $i$.}.
        \end{definition}
        
        

Positive USGs (PUSGs) have many interesting properties that general USGs do not necessarily possess. Before we state these properties, we need the following lemma from linear algebra.
        \begin{lemma}
        \label{Perron-Frobenius}
            (Perron-Frobenius \cite{BP79}): For any square matrix $A > 0$, we have
            \begin{itemize}
                \item $A$ has an eigenvalue $\lambda > 0$. Moreover, for any other eigenvalue $\mu$ of $A$, $|\lambda| > |\mu|$. We call $\lambda$ the Perron-Frobenius value, or spectral radius of $A$, denoted as $\lambda = \rho(A)$.
                \item The eigenvalue $\lambda$ has algebraic and geometric multiplicity one.
                \item There is an eigenvector $x > 0$ of $A$ with an eigenvalue of $\lambda$.
                \item The only positive eigenvectors of $A$ have the form $k x$ for some $k > 0$. Moreover, all positive eigenvectors have corresponding eigenvalue $\lambda$.
            \end{itemize}
        \end{lemma}
        

        \begin{lemma}
        \label{share}
            For payoff matrices $A > 0$, $B > 0$, $AB$ and $BA$ share at least one positive eigenvalue, which is their spectral radius.
        \end{lemma}
        \begin{proof}
           Clearly, $AB$ and $BA$ are square matrices. Let $x > 0$ be an eigenvector of $AB$ with eigenvalue $\lambda = \rho(AB) > 0$, whose existence is guaranteed by Lemma~\ref{Perron-Frobenius}. Note that
            $$
                BA(Bx) = B(ABx) = \lambda(Bx).
            $$
            Namely, $Bx$ is an eigenvector of $BA$ with eigenvalue $\lambda$. It follows that $AB$ and $BA$ share the same positive eigenvalue $\lambda > 0$. Now suppose $\rho(BA) > \lambda$. By the same argument, we can see that $\rho(BA)$ is an eigenvalue of $AB$, a contradiction.
        \end{proof}
        
            With Lemma~\ref{share}, we are now able to derive two NEs for all PUSGs.
        \begin{theorem}
        \label{thm:4}
            There exists two NE $(x_1, y_1)$, $(x_2, y_2)$ for any PUSG, where
            \begin{itemize}
                \item $x_1 > 0$ is the unit eigenvector of $AB$ with eigenvalue $\lambda = \rho(AB)$.
                \item $y_1 = \frac{Bx_1}{\|Bx_1\|_2}$.\\
                where the utilities of the players obtained from $(x_1, y_1)$ are $\left(\frac{\lambda}{\|B x_1\|_2}, \|B x_1\|_2\right)$.
                \item $y_2 > 0$ is the unit eigenvector of $BA$ with eigenvalue $\lambda = \rho(BA)$.\footnote{Recall that $\rho(AB) = \rho(BA)$}
                \item $x_2 = \frac{Ay_2}{\|Ay_2\|_2}$.\\
                where the utilities of the players obtained from $(x_2, y_2)$ are $\left(\|A y_2\|_2, \frac{\lambda}{\|A y_2\|_2}\right)$.
            \end{itemize}
        \end{theorem}
        \begin{proof}
            We prove for the case of $(x_1, y_1)$. The case of $(x_2, y_2)$ is symmetric. By Lemma~\ref{share}, it is always feasible to pick $x_1$ as stated in the theorem. For player 1,
            \begin{align*}
                u_1(x', y_1) = &\ x'^T A y_1 = \frac{1}{\|B x_1\|_2} x'^T A B x_1 \\
                = &\ \frac{\lambda}{\|B x_1\|_2} x'^T x_1 \le \frac{\lambda}{\|B x_1\|_2} x_1^T x_1 \\
                = &\ \frac{\lambda}{\|B x_1\|_2}.
            \end{align*}
            For player 2,
            \begin{align*}
                u_2(x_1, y') = &\ y'^T B x_1 = y'^T \|B x_1\|_2 y_1 \\
                \le &\ \|B x_1\|_2 y_1^T y_1 = \|B x_1\|_2.
            \end{align*}
            In other words, neither player has profitable deviation in $(x_1, y_1)$.
        \end{proof}
        
            Theorem~\ref{thm:4} derives a pair of symmetric NEs for any PUSG. One might wonder whether the two NEs are identical? This is indeed the case. We dedicate Subsection~\ref{sec: unique} to this result.
        
            In fact, there is a symmetric NE in a PUSG if the payoff matrices satisfy certain additional conditions. Before we state these conditions, we need the following technical lemma.
        \begin{lemma}
            For square matrices $A > 0$, $B > 0$ such that $AB = BA$, $A$ and $B$ share the same one-dimensional eigenspace of spectral radius.
        \end{lemma}
        \begin{proof}
            Let $\lambda = \rho(A)$, $x > 0$ be an eigenvector of $A$ whose corresponding eigenvalue is $\lambda$, then
            $$
                A (B x) = B (A x) = \lambda (B x),
            $$
            namely $Bx$ is an eigenvector of $A$ whose eigenvalue is $\lambda$. By Lemma~\ref{Perron-Frobenius}, the eigenspace of $\lambda$ is one-dimensional, which implies that $B x = \mu x$ for some $\mu$. Again by Lemma~\ref{Perron-Frobenius}, $x$ belongs to the eigenspace of the spectral radius of $B$, or equivalently $\mu = \rho(B)$.
        \end{proof}
        
        If $AB = BA$, the corresponding PUSG has a symmetric NE.
        
        \begin{theorem}
        \label{thm:3.3}
            There is a symmetric NE $(x, x)$ for any PUSG with square payoff matrices $A\times B$ such that $AB = BA$. The NE utilities are $(\rho(A), \rho(B))$.
        \end{theorem}
        \begin{proof}
            Let $x > 0$ be the unit eigenvector of $A$ whose corresponding eigenvalue is $\rho(A)$ (and therefore the unit eigenvector of $B$ whose eigenvalue is $\rho(B)$). For player 1,
            \begin{align*}
                u_1(x', x) = &\ x'^T A x = \rho(A) x'^T x \\
                \le &\ \rho(A) x^T x = \rho(A).
            \end{align*}
            For player 2,
            \begin{align*}
                u_2(x', x) = &\ x'^T B x = \rho(B) x'^T x \\
                \le &\ \rho(B) x^T x = \rho(B).
            \end{align*}
            Neither player has a profitable deviation in $(x,x)$.
        \end{proof}
    \subsection{Uniqueness of NE in PUSGs}
    \label{sec: unique}
        
        One of the most appealing properties of all PUSGs is that they have unique NE.
        \begin{theorem}
        \label{thm:3.4}
            Any PUSG has an unique NE.
        \end{theorem}
        \begin{proof}
           Let $(x, y)$ be an arbitrary NE of PUSG with payoff matrices $A$ and $B$, whose existence has been established in Theorem~\ref{thm:4}. By Lemma~\ref{lemma:1},
            $$
                \exists \lambda > 0,\ \mu > 0,\ \mathrm{s.t.}\ ABx = \lambda x,\ BAy = \mu y
            $$
            We will show that $\lambda$ is the spectral radius of $AB$, and $x$ is the corresponding positive unit eigenvector. The case of $y$ is symmetric. Assume $\lambda \ne \rho(AB)$. By Lemma~\ref{Perron-Frobenius}, there must be some $i \in [n]$ such that $x_i = 0$, since there are no other positive eigenvectors beside those of the spectral radius. Note that $\lambda > 0$, $AB > 0$.
            \begin{align*}
                0 & = \lambda x_i = (ABx)_i = \sum_j (AB)_{ij} x_j \\
                & \ge \min (AB)_{ij} \|x\|_1 > 0,
            \end{align*}
            a contradiction. Therefore $\lambda = \rho(AB)$. Again by Lemma~\ref{Perron-Frobenius}, the eigenspace of $\lambda$ is one-dimensional. Namely $x$ is the unique positive eigenvector of $\lambda$ such that $\|x\|_2 = 1$. The same argument works for $y$. To conclude, we prove that $(x, y)$ is the unique NE.
        \end{proof}

        \begin{corollary}
            Any PUSG has an unique NE, which has the form stated in Theorem~\ref{thm:4}. Moreover, the two symmetric NEs in Theorem~\ref{thm:4} are identical.
        \end{corollary}
        
            Next, we show the unique NE of a PUSG can be efficiently found via a natural learning process.
\section{Solving PUSGs via Cournot adjustments}
    
        In this section, we show that the unique NE of any PUSG can be resulted when both players follow a well-known learning process called {\em Cournot adjustments
        }. This is remarkable property since it states that players can learn to play NE even without any information of each other's payoff matrix.
        
    \subsection{Cournot adjustments}
        
Define {\em Cournot adjustments} as follows,

            \begin{enumerate}
                \item In the first round, each player $i$ plays any positive strategy $s_i^0 > 0$.
                \item In round $t$, each player $i$ observes $s_{-i}^t$, the strategy of player $-i$.
                \item In round $t+1$, each player $i$ plays her best response against $s_{-i}^t$. Namely
                    $$
                        s_i^{t + 1} = \frac{A_i s_{-i}^t}{\|A_i s_{-i}^t\|_2}.
                    $$
                \item Iterate until no player updates her strategy.
            \end{enumerate}
            
Cournot adjustments define a natural protocol for players to learn to play a game over time. It is appealing when players do not know others' payoff matrices and for whatever reason that the players cannot perform equilibrium computation upfront. It is known that, for any standard games, a carefully designed better response dynamics can converge to some mixed-strategy Nash equilibrium (aka. Nash's proof), but may take exponential number of rounds. In the following, we show that this procedure thoroughly exploits the properties of PUSGs and finds efficiently the unique NE for any PUSG in logarithmic number of rounds with respect to initial error.
      

    \subsection{Convergence of Cournot adjustments in PUSGs}
        
        To formally state and prove the convergence result, we need the following proposition from numerical analysis.
        \begin{lemma}
        \label{convergence}
            (Convergence of power iteration \cite{MP29}): For any positive square matrix $A$ whose eigenvalue with the largest modulus is $\lambda$ and the corresponding eigenspace is $E$, let $x_0$ be an arbitrary unit vector such that $x$ is not orthogonal to $E$. Let
            $$
                x^t = \frac{A x^{t - 1}}{\|A x^{t - 1}\|_2}.
            $$
            It is guaranteed that $x^t$ converges to $x^*$, where $A x^* = \lambda x^*$. Moreover,
            $$
                \forall p \in \mathbb{Z}^+ \cup \{\infty\},\ \exists r \in (0, 1), c \in \mathbb{R}^+,\ \mathrm{s.t.}
            $$
            $$
                \|x^t - x^*\|_p \le c r^t.
            $$
        \end{lemma}
        
         In presence of Lemma~\ref{convergence}, we now state a convergence result of best response dynamics in PUSGs.
        \begin{theorem}
           If both players follow Cournot adjustment, the strategy sequence $(x^t, y^t)$ converges to the unique NE of the PUSG with an exponentially decreasing error. Or in other words, there is a linear convergence\footnote{Linear convergence is another way of saying the error diminishes exponentially fast in the number of iterations.} for $(x^t, y^t)$ to the NE, following Cournot adjustment.
        \end{theorem}
        \begin{proof}
            Let $A$ and $B$ be the payoff matrices. We can explicitly derive the strategy expressions of Cournot adjustments in round $t$ as follows,
            $$
                x^t = A y^{t - 1},\ y^t = B x^{t - 1}.
            $$
            It follows that
            $$
                x^{2k} = (AB)^k x^0,\ y^{2k} = (BA)^k y^0,\ \forall k \in \mathbb{N}.
            $$
            Since we choose $x^0 > 0$, $y^0 > 0$, by Lemma~\ref{Perron-Frobenius}, it is impossible that $x^0$ (resp.\ $y^0$) is orthogonal to the eigenspace of the spectral radius of $AB$ (resp.\ $BA$). By Lemma~\ref{convergence}, as $k$ grows, $x^{2k}$ converges to the positive unit eigenvector of $AB$ exponentially fast, and $y^{2k}$ converges to that of $BA$. Therefore $(x^{2k}, y^{2k})$ converges to the unique PSNE exponentially fast. As $(x^{2k}, y^{2k})$ converges, $(x^{2k + 1}, y^{2k + 1})$ converges as well, concluding the proof.
        \end{proof}

\section{Approximating mixed-strategy equilibrium in standard games via USGs}

It is well-known that computing a mixed-strategy Nash equilibrium (MSNE) in standard two-player games is PPAD-complete \cite{CD06}. In this section, we show that our understanding of USG can help us to compute an approximate MSNE for any standard games.    

    \subsection{Approximation Scheme}
        
Consider any PUSG. By theorems we have derived so far, one can easily compute the unique NE $(x, y)$ of the PUSG. We now normalize $x$ and $y$ to be $x'$ and $y'$, so that $\|x'\|_1 = \|y'\|_1 = 1$. Our main finding is that $(x', y')$ is a multiplicative $O\left(\sqrt{\max(m, n)}\right)$-approximate MSNE\footnote{A {\em multiplicative $k$-approximate MSNE} denotes a strategy profile where no player can improve her utility by $k$ times via deviation.} for the underlying standard two-player game. 
        
           Call this approximation scheme the {\em simple approximate scheme}.
    \subsection{Approximation via simple approximate scheme}
        
   Once again, before we state and prove our result, we need the following technical lemma.
        \begin{lemma}
        \label{approx}
            $$
                \min_{x \in \mathbb{R}^n,\ \|x\|_1 = 1} \left\{\frac{\|x\|_2^2}{\|x\|_\infty}\right\} = \frac{2}{\sqrt{n} + 1}. 
            $$
        \end{lemma}
        \begin{proof}
        
            Let $t = \|x\|_\infty \ge \frac{1}{n}$. Obviously,
            \begin{align*}
                \frac{\|x\|_2^2}{\|x\|_\infty} & \ge \frac{t^2 + (n - 1) \left(\frac{1 - t}{n - 1}\right)^2}{t} \\
                & = \frac{n}{n - 1} t - \frac{2}{n - 1} + \frac{1}{t(n - 1)} \\
                & \ge \frac{2}{\sqrt{n} + 1}.
            \end{align*}

        \end{proof}
        
            We are now ready to state our main result of the section.
        \begin{theorem}
            For any standard two-player game with payoff matrices $A$ and $B$, the simple approximation scheme yields a multiplicative $O\left(\sqrt{\max(m, n)}\right)$-approximate MSNE, where $m$ is the number of rows of $A$, and $n$ is the number of rows of $B$.
        \end{theorem}
        \begin{proof}
            Let $(x, y)$ be the NE of the induced PUSG over payoff matrices $A\times B$, and $(x', y')$ be the normalized vectors, as stated in the simple scheme. Since $(x, y)$ is an NE in the PUSG, $\exists \lambda$, $\mu$, s.t.\
            $$
                Ay' = \lambda x',\ Bx' = \mu y'.
            $$
            Consider player one's payoff with or without deviation. 
            
            Without deviation, she gets
            $$
                u_1(x', y') = x'^T A y' = \lambda \|x'\|_2^2.
            $$
            
            By deviation, she gets
            \begin{align*}
                & \max_{\|x_1\|_1 = 1} u_1(x_1, y') = \max_{\|x_1\|_1 = 1} x_1^T A y' \\
                = &\ \lambda \max_{\|x_1\|_1 = 1} x_1^T x' = \lambda \|x'\|_\infty.
            \end{align*}
            By Lemma~\ref{approx},
            \begin{align*}
                & \min_{(x',y')} \frac{u_1(x', y')}{\max_{\|x_1\|_1 = 1} u_1(x_1, y')} \\
                = &\ \min_{(x',y')} \frac{\|x'\|_2^2}{\|x'\|_\infty} \ge \frac{1}{\sqrt{m}} + \frac{m - 2\sqrt{m} + 1}{\sqrt{m}(m - 1)} \\
                = &\ \Omega\left(\frac{1}{\sqrt{m}}\right).
            \end{align*}
            Symmetrically, for player two, the approximate factor becomes $\Omega\left(\frac{1}{\sqrt{n}}\right)$.
        \end{proof}

\section{Multiplayer PUSGs}
\label{sec:n-player}

    \begin{definition}
        An $m$-player PUSG is defined as $(A^1, \dots, A^m)$, where $A^k = \left(A_{i_1, \dots, i_m}^k\right)$ is the game tensor for player $k$, such that $A_{i_1, \dots, i_m}^k > 0$ for all $i_1 \in [n_1], \dots, i_m \in [n_m]$.
    \end{definition}

\subsection{Existence of NE in multiplayer PUSGs}

    \begin{lemma}
    \label{Brouwer}
        (Brouwer's fixed point theorem): For any $n \in \mathbb{Z}^+$, $\Omega \subseteq \mathbb{R}^n$ which is compact and convex, $f: \Omega \rightarrow \Omega$ which is continuous, there is some $x^* \in \Omega$ such that $f(x^*) = x^*$.
    \end{lemma}

    \begin{theorem}
        There exists an NE for any $m$-player PUSG $(A^1, \dots, A^m)$.
    \end{theorem}

        The proof resembles that of the existence of MSNE in normal form games.

    \begin{proof}
        Let
        $$
            s_i = \sum_{j \le i} n_j,
        $$
        $$
            \mathbb{R}^{s_m} \supset \Omega = \left\{x \in \mathbb{R}^{s_m} \Big| x_{i, j} \ge 0,\,\forall i \in [m],  j \in [n_i],\, \|x_i\|_1 = 1 \,\forall i \in [m]\right\}.
        $$
        For all $x = (x_1, \dots, x_m) \in \Omega$,
        $$
            f(x) = \left(\frac{A^1 x_2 x_3 \dots x_m}{\|A^1 x_2 x_3 \dots x_m\|_1}, \frac{A^2 x_1 x_3 \dots x_m}{\|A^2 x_1 x_3 \dots x_m\|_1}, \dots, \frac{A^m x_1 x_2 \dots x_{m - 1}}{\|A^m x_1 x_2 \dots x_{m - 1}\|_1}\right).
        $$
        It is easy to verify that $\Omega$ and $f$ satisfy the conditions in Lemma~\ref{Brouwer}. Therefore there is some $x^* = (x_1^*, \dots, x_m^*)$ satisfying $f(x^*) = x^*$, which implies that there is some $\lambda_i$ such that $A^i x_1^* \dots x_m^* = \lambda_i x_i^*$ for all $i$. So $\left(\frac{x_1^*}{\|x_1^*\|_2}, \dots, \frac{x_m^*}{\|x_m^*\|_2}\right)$ is an MSNE of the PUSG.
    \end{proof}
    
    Note that, this theorem also follows from~\cite[Theorem 1]{Fiat10}.

\subsection{Subclasses of multiplayer PUSGs}

In this subsection, we investigate several subclasses of multiplayer PUSGs where NE is easy to find.

    \subsubsection{Symmetric PUSGs with even number of players}

        First we give an algorithm that solves $m$-player symmetric PUSGs when $m$ is even.

        \begin{definition}
            An $m$-player symmetric PUSG is a PUSG where $A^i = A^j$ for all $i, j \in [m]$, and $A_{i_1, \dots, i_m}^k = A_{\sigma(i_1), \dots, \sigma(i_m)}^k$ for all $k \in [m]$ and $\sigma \in S_n$, where $S_n$ is the permutation group and $n$ is the number of actions of each player.
        \end{definition}
        
        The method used to find NE is symmetric PUSG is called {\em SS-HOPM}. SS-HOPM outputs a symmetric NE with a particular payoff (which equals the largest Z-eigenvalue of the payoff tensor). The linear convergence of SS-HOPM has been shown originally in \cite{KM11} and revised in \cite{CPZ13}.

        The algorithm performs as the following.
        \begin{enumerate}
            \item Choose $x^0 > 0$, and the shift constant $\alpha = \lceil m \sum_{i_1, \dots, i_m} A_{i_1, \dots, i_m} \rceil$.
            \item Let $y^{t + 1} = A(x^t)^{m - 1} + \alpha x^t$.
            \item Compute
                $$
                    x^{t + 1} = \frac{y^{t + 1}}{\|y^{t + 1}\|_2},\ \lambda^{t + 1} = A(x^{t + 1})^m.
                $$
        \end{enumerate}

        As shown in \cite{CPZ13}, $x^t$ converges to an symmetric PSNE $x^*$ while $\lambda^t$ converges to the payoff of each player playing $x^*$.

    \subsubsection{Markov PUSGs via Cournot Adjustments}

        Generalizing \cite{LN14}, we show that a unique PSNE exists in any {\em Markov PUSG}, which can be efficiently achieved via Cournot adjustments.

        \begin{definition}
            An Markov PUSG $(A^1, \dots, A^m)$ is a PUSG such that
            $$
                \sum_{i_k} A_{i_1, \dots, i_m}^k = c_k
            $$
            for all $k \in [m], i_1 \in [n_1], \dots, i_{k - 1} \in [n_{k - 1}], i_{k + 1} \in [n_{k + 1}], \dots, i_m \in [n_m]$ and a constant $c_k$.
        \end{definition}

In words, Markov PUSG is a subset of PUSG such that, fixing any other players' strategy profile, the sum of player $k$'s utility over all his/her actions is a constant, for any $k$.
        Since every Markov PUSG can be scaled such that for all $k$, $c_k = 1$, it is without loss of generality to consider only those games.

        \begin{lemma}
        \label{conservation}
            For nonnegative $x_1, \dots, x_m$ such that $\|x_i\|_1 = 1$, we have $\|A^k x_1 \dots x_{k - 1} x_{k + 1} \dots x_{n}\|_1 = 1$ for all $k \in [m]$.
        \end{lemma}
        \begin{proof}
            \begin{align*}
                & \|A^k x_1 \dots x_{k - 1} x_{k + 1} \dots x_{n}\|_1 \\
                =\ & \sum_{i_k} \sum_{i_1, \dots, i_{k - 1}, i_{k + 1}, \dots, i_m} A_{i_1, \dots, i_m}^k x_{1, i_1} \dots x_{k - 1, i_{k - 1}} x_{k + 1, i_{k + 1}} \dots x_{m, i_m} \\
                =\ & \sum_{i_1, \dots, i_{k - 1}, i_{k + 1}, \dots, i_m} \sum_{i_k} A_{i_1, \dots, i_m}^k x_{1, i_1} \dots x_{k - 1, i_{k - 1}} x_{k + 1, i_{k + 1}} \dots x_{m, i_m} \\
                =\ & \sum_{i_1, \dots, i_{k - 1}, i_{k + 1}, \dots, i_m} x_{1, i_1} \dots x_{k - 1, i_{k - 1}} x_{k + 1, i_{k + 1}} \dots x_{m, i_m} \\
                =\ & 1
            \end{align*}
        \end{proof}

        \begin{lemma}
        \label{m_convergence}
            Let $\Omega$ be as defined above, $f: \Omega \rightarrow \Omega$ be such that for $v \in \Omega$,
            $$
                f(v)_k = A^k v_1 \dots v_{k - 1} v_{k + 1} \dots v_m.
            $$
            For $x = (x_1, \dots, x_m) \in \Omega$, $y = (y_1, \dots, y_m) \in \Omega$,
            $$
                \|f(x)_k - f(y)_k\|_1 \le (1 - \delta_k) \left(\sum_{i \in [m],\, i \ne k} \|x_i - y_i\|_1 \right),
            $$
            where
            $$
                \delta_k = \min_{V \subseteq [n_k]} \left[\min_{i_1, \dots, i_{k - 1}, i_{k + 1}, \dots, i_m} \sum_{i_k \in V} A_{i_1, \dots, i_m}^k + \min_{i_1, \dots, i_{k - 1}, i_{k + 1}, \dots, i_m} \sum_{i_k \in V'} A_{i_1, \dots, i_m}^k\right],
            $$
            and $V' = [n_k] \setminus V$.
        \end{lemma}
        \begin{proof}
            \begin{align*}
                & \sum_{i_k \in V_k} (f(x)_{k, i_k} - f(y)_{k, i_k}) \\
                =\ & \sum_{i_k \in V_k} \sum_{i_1, \dots, i_{k - 1}, i_{k + 1}, \dots, i_m} A_{i_1, \dots, i_m} (x_{1, i_1} \dots, x_{k - 1, i_{k - 1}}, x_{k + 1, i_{k + 1}}, \dots, x_{m, i_m} - \\
                & y_{1, i_1} \dots, y_{k - 1, i_{k - 1}}, y_{k + 1, i_{k + 1}}, \dots, y_{m, i_m}) \\
                =\ & \sum_{i_k \in V_k} \sum_{i_1, \dots, i_{k - 1}, i_{k + 1}, \dots, i_m} A_{i_1, \dots, i_m} [(x_{1, i_1} - y_{1, i_1})x_{2, y_2} \dots, x_{m, i_m} + \\
                & y_{1, i_1}(x_{2, i_2} - y_{2, i_2})x_{3, i_3} \dots x_{m, i_m} + \dots + y_{1, i_1} \dots y_{m - 1, i_{m - 1}} (x_{m, i_m} - y_{m, i_m})].
            \end{align*}
            Let $V_k \subseteq [n_k]$ be the largest set such that $\forall i_k \in V_k$, $f(x)_{k, i_k} > f(y)_{k, i_k}$, $V_1 \subseteq [n_1]$ the largest set such that $\forall i_1 \in V_1$, $x_{1, i_1} > y_{1, i_1}$. Note that by Lemma~\ref{conservation}, $\|x_k\|_1 = \|y_k\|_1 = 1$, and hence $\sum_{i_k} x_{k, i_k} - y_{k, i_k} = 0$ for all $k \in [m]$. We then have
            \begin{align*}
                & \sum_{i_1, \dots, i_{k - 1}, i_{k + 1}, \dots, i_m} \sum_{i_k \in V_k} A_{i_1, \dots, i_m} (x_{1, i_1} - y_{1, i_1}) y_{2, i_2} \dots, y_{m, i_m} \\
                =\ & \sum_{i_1 \in V_1} \sum_{i_2, \dots, i_{k - 1}, i_{k + 1}, \dots, i_m} \sum_{i_k \in V_k} A_{i_1, \dots, i_m} (x_{1, i_1} - y_{1, i_1}) y_{2, i_2} \dots, y_{m, i_m} + \\
                & \sum_{i_1 \notin V_1} \sum_{i_2, \dots, i_{k - 1}, i_{k + 1}, \dots, i_m} \sum_{i_k \in V_k} A_{i_1, \dots, i_m} (x_{1, i_1} - y_{1, i_1}) y_{2, i_2} \dots, y_{m, i_m} \\
                \le\ & \sum_{i_1 \in V_1} \sum_{i_2, \dots, i_{k - 1}, i_{k + 1}, \dots, i_m} \left(\max_{j_1 \in V_1, j_2, \dots, j_m} \sum_{j_k \in V_k} A_{j_1, \dots, j_m}\right) (x_{1, i_1} - y_{1, i_1}) y_{2, i_2} \dots, y_{m, i_m} - \\
                & \sum_{i_1 \notin V_1} \sum_{i_2, \dots, i_{k - 1}, i_{k + 1}, \dots, i_m} \left(\min_{j_1 \in V_1, j_2, \dots, j_m} \sum_{j_k \in V_k} A_{j_1, \dots, j_m}\right) (y_{1, i_1} - x_{1, i_1}) y_{2, i_2} \dots, y_{m, i_m} \\
                =\ & \left(\max_{j_1 \in V_1, j_2, \dots, j_m} \sum_{j_k \in V_k} A_{j_1, \dots j_m} - \min_{j_1 \in V_1, j_2, \dots, j_m} \sum_{j_k \in V_k} A_{j_1, \dots, j_m}\right) \times \\
                & \sum_{i_1 \notin V_1} \sum_{i_2, \dots, i_{k - 1}, i_{k + 1}, \dots, i_m} (y_{1, i_1} - x_{1, i_1})y_{2, i_2} \dots y_{k - 1, i_{k - 1}} y_{k + 1, i_{k + 1}} \dots y_{m, i_m} \\
                \le\ & \left(1 - \min_{j_1, \dots, j_m} \sum_{j_k \notin V_k} A_{j_1, \dots, j_m} - \min_{j_1, \dots, j_m} \sum_{j_k \in V_k} A_{j_1, \dots, j_m}\right) \sum_{i_1 \notin V_1} (y_{1, i_1} - x_{1, i_1}) \\
                \le\ & \frac12 (1 - \delta_k) \|x_1 - y_1\|_1.
            \end{align*}
            We therefore get
            \begin{align*}
                \|f(x)_k - f(y)_k\|_1 & = 2 \sum_{i_k \in V_k} (f(x)_{k, i_k} - f(y)_{k, i_k}) \\
                & \le 2 \sum_{i \ne k} \left[\frac12 (1 - \delta_k) \|x_i - y_i\|_1\right] \\
                & = (1 - \delta_k) \left(\sum_{i \ne k} \|x_i - y_i\|_1\right).
            \end{align*}
        \end{proof}

        \begin{theorem}
        \label{m_uniqueness}
            There exists an unique NE in any Markov PUSG where $\delta_k > \frac{m - 2}{m - 1}$ for all $k$.
        \end{theorem}
        \begin{proof}
            Assume there are two distinct NE in an $m$-player game $(A^1, \dots, A^m)$, $x_0$ and $y_0$. Let $x = \frac{x_0}{\|x_0\|_1}$, $y = \frac{y_0}{\|y_0\|_1}$. By Lemma~\ref{m_convergence},
            \begin{align*}
                & \|x - y\|_1 \\
                \ = & \sum_k \|x_k - y_k\|_1 \\
                \ = & \sum_k \|f(x)_k - f(y)_k\|_1 \\
                \ \le & \sum_k \sum_{i \ne k} (1 - \delta_k) \left(\|x_i - y_i\|_1\right) \\
                \ < & \sum_{i \in [m]} (m - 1) \left(1 - \frac{m - 2}{m - 1}\right) \left(\|x_i - y_i\|_1\right) \\
                \ = & \|x - y\|_1,
            \end{align*}
            an contradiction.
        \end{proof}

        \begin{theorem}
            Cournot adjustments lead to the unique NE in any Markov PUSG where $\delta_k > \frac{m - 2}{m - 1}$ for all $k$.
        \end{theorem}
        \begin{proof}
            For simplicity, we denote strategies by vectors whose $L_1$-norm are scaled to $1$ in the proof. Consider a procedure where player $k$ starts by playing $x_k^0 = (\frac{1}{n_k}, \dots, \frac{1}{n_k})$. Let $x^*$ be the unique PSNE of the game, guaranteed to exist by Theorem~\ref{m_uniqueness}. Let $\epsilon_0 = \max_{i \in [m]} \|x_i^0 - x_i^*\|_1$, $\delta = \max_{i \in [m]} (1 - \delta_i)$. Clearly, in round $t$, strategies of player $k$ will be $x_k^t = f(x^{t - 1})_k$, where $f$ is the same as stated above.

            On the other hand, as shown in Lemma~\ref{m_convergence},
            $$
                \|x_k^t - x_k^*\|_1 \le \delta \left(\sum_{i \ne k} \|x_i^{t - 1} - x_i^*\|_1\right),\,\forall t \in \mathbb{Z}^+,\,k \in [m].
            $$
            By a simple induction, we prove that
            $$
                \epsilon_t = \max_k \left[\|x_k^t - x_k^*\|_1\right] \le (m - 1)^t \delta^t \epsilon_0.
            $$
            When $t = 0$, it holds obviously that $\epsilon_0 \le \epsilon_0$. Assume that $\epsilon_{t - 1} \le (m - 1)^{t - 1} \delta^{t - 1} \epsilon_0$, we may show,
            $$
                \epsilon_t = \max_k \left[\|x_k^t - x_k^*\|_1\right] \le \max_k \delta \left(\sum_{i \ne k} \|x_i^{t - 1} - x_i^*\|_1\right) \le \delta (m - 1) \epsilon_{t - 1} \le (m - 1)^t \delta^t \epsilon_0.
            $$
            It can be seen easily that $\epsilon_t$ goes to $0$ exponentially fast considering that $\delta < \frac{1}{m - 1}$.

            We have shown that the $L_1$-norm normalized strategies converge to the $L_1$-norm normalized NE. It follows naturally that the strategies themselves converge to the unique NE. Morevoer, the convergence is linear, i.e., the error decreases exponentially fast.
        \end{proof}

\subsection{Multiplicity of NE in multiplayer USGs}
        
In fact, there may be infinitely many NEs in a multiplayer USG. Here is an interesting example~\cite{CPZ13}.

\begin{example}
    Consider a 4-player USG where game tensors $(A^1, A^2, A^3, A^4)$ are such that $A_{1112}^1 = A_{2122}^1 = A_{1112}^2 = A_{1222}^2 = A_{1112}^3 = A_{2122}^3 = A_{1121}^4 = A_{2122}^4 = 2$, $A_{i_1 i_2 i_3 i_4}^j = 0$ otherwise. We consider symmetric strategy $x = (x_1, x_2)$. In order for $x$ to be an NE, we need
    $$
        \left\{
        \begin{array}{c}
            2 x_1^2 x_2 = \lambda x_1 \\
            2 x_1 x_2^2 = \lambda x_2 \\
            x_1^2 + x_2^2 = 1
        \end{array}
        \right..
    $$
    By setting $\lambda = 2 x_1 x_2$, it appears that any pair of $(x_1, x_2)$ where $x_1^2 + x_2^2 = 1$ forms a symmetric NE. Moreover, any equilibrium payoff $\lambda \in [0, 1]$ can be achieved by some choice of $(x_1, x_2)$.
\end{example}

\section{Acknowledgements}

We are grateful to Andrew Yao and Christos Papadimitriou for helpful discussions. This work was supported by the National Basic Research Program of China Grant 2011CBA00300, 2011CBA00301, the Natural Science Foundation of China Grant 61033001, 61361136003, 61303077, Tsinghua university Initiative Scientific Research Grant and China Youth 1000-talent program.

%
%
%
%
%
%

\bibliographystyle{plain}
\bibliography{redistribution_refs}

\end{document}